\documentclass[]{article}
\usepackage{graphicx}

\usepackage{amssymb}
\usepackage{amsmath}
\usepackage{amsthm}

\usepackage{lineno}

\newtheorem{thm}{Theorem}[section]
\newtheorem{lem}[thm]{Lemma}

\theoremstyle{definition}

\begin{document}

\title{On the distance preserving trees in graphs}
\author{Toru Araki \and Shingo Osawa \and Takashi Shimizu}

\maketitle

\begin{abstract}
  For a vertex $v$ of a graph $G$, a spanning tree $T$ of $G$ is
  distance-preserving from $v$ if, for any vertex $w$, the distance
  from $v$ to $w$ on $T$ is the same as the distance from $v$ to $w$
  on $G$.
  If two vertices $u$ and $v$ are distinct, then two
  distance-preserving spanning trees $T_{u}$ from $u$ and $T_{v}$ from
  $v$ are distinct in general.
  A purpose of this paper is to give a characterization for a given
  weighted graph $G$ to have a spanning tree $T$ such that $T$ is a
  distance-preserving spanning tree from distinct two vertices.
\end{abstract}


\section{Introduction}
\label{sec:introduction}

Let $G$ be a simple undirected graph.
The vertex set and the edge set of $G$ is denoted by $V(G)$ and
$E(G)$, respectively.
For a subset $U \subseteq V(G)$, the subgraph induced by $U$ is
denoted by $G[U]$.
A \emph{weighted graph} is a graph each edge of whose edges is
assigned a real number (called the \emph{cost} or \emph{weight} of the
edge).
We denote the weight of an edge $e$ of $G$ by $w(e)$.
For a path $P$ of $G$, the \emph{length} of $P$ is defined as the sum
of the weights of its edges.
The \emph{distance} between two vertices $u$ and $v$ of a graph $G$ is
the minimum length of paths from $u$ to $v$, and is denoted by
$d_{G}(u,v)$.
For a subset of vertices $S$, the distance from $u$ to $S$ is defined
by
\begin{displaymath}
  d_{G}(u,S) = \min_{v \in S} d_{G}(u,v).
\end{displaymath}

Let $v$ be a vertex of $G$.
A spanning tree $T$ of $G$ is a \emph{distance-preserving spanning
  tree} (or a \emph{DP-tree} for short) from $v$ if $d_{T}(v,w) =
d_{G}(v,w)$ for each $w \in V(G)$.
An example of a DP-tree $T$ from $u$ in a graph $G$ is shown in
Fig.~\ref{fig:dp-tree-example}.

In a well-known book ``Graphs and Digraphs'' written by Chartrand,
Lesniak, and Zhang~\cite{chartrand11}, we can find an exercise of
Section 2.3: ``Give an example of a connected graph $G$ that is not a
tree and two vertices $u$ and $v$ of $G$ such that a
distance-preserving spanning tree from $v$ is the same as a
distance-preserving spanning tree from $u$.''
In Fig.~\ref{fig:dp-tree-example}, the spanning tree $T$ is
distance-preserving from the two vertices $u$ and $v$.
Hence Fig.~\ref{fig:dp-tree-example} is an answer the question.

A purpose of this paper is to give a complete answer of the question.
That is, for a given weighted graph $G$ and two vertices $u$ and $v$,
we would like to give a characterization for a graph $G$ to have a
spanning tree $T$ such that $T$ is a distance-preserving spanning tree
from $u$ as well as from $v$.

\begin{figure}[tb]
  \centering
  \includegraphics{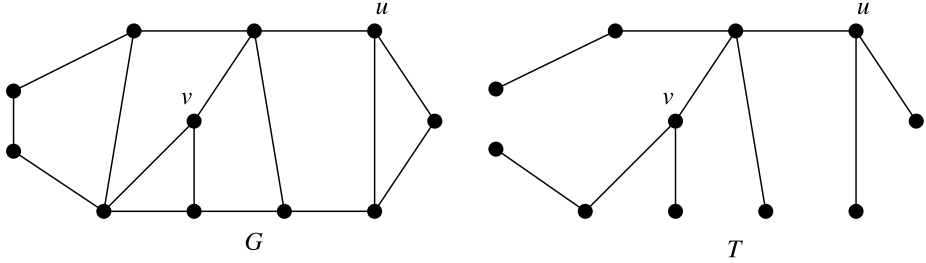}
  \caption{An example of a distance-preserving tree $T$ from $u$ of a
    graph $G$.  We assume that the weights of edges are 1.  The tree
    $T$ is also a distance-preserving tree from $v$.}
  \label{fig:dp-tree-example}
\end{figure}


\section{Main result}
\label{sec:main-result}

In this section, we show the following theorem.
If a spanning tree $T$ of $G$ is a distance-preserving spanning tree
from $u$ as well as from $v$, we say that $T$ is a \emph{common
  distance-preserving spanning tree} of $u$ and $v$ in $G$.

\begin{thm}
  \label{thm:main}

  Let $G$ be a weighted graph and $u$ and $v$ are two vertices of $G$.
  A spanning tree $T$ of $G$ is a common distance-preserving spanning
  tree of $u$ and $v$ if and only if the following three conditions
  hold.

  \begin{enumerate}
  \item[(1)] A shortest $u$-$v$ path $P$ in $G$ is unique.

  \item[(2)] We define the unique shortest $u$-$v$ path as
    $P=(u=v_{0},v_{1},\dots,v_{k}=v)$.
    For any vertex $x$, there is a unique vertex $v_{i} \in V(P)$ such
    that $d_{G}(x,v_{i})=d_{G}(x,V(P))$.

  \item[(3)] For $0 \leq i \leq k$, let $V_{i}=\{x \mid x \in V(G) \text{
      and } d_{G}(x,v_{i})=d_{G}(x,V(P))\}$.
    If $e = xy \in E(G)$ for $x \in V_{i}$ and $y \in V_{j}$, then
    $w(e) \geq d_{G}(v_{i},v_{j})$ and
    $|d_{G}(v_{i},x)-d_{G}(v_{j},y)| \leq w(e)-d_{G}(v_{i},v_{j})$.
  \end{enumerate}
\end{thm}

We first show the necessary condition of Theorem~\ref{thm:main}.

\begin{lem}
  \label{lem:nec_one}
  If $u$ and $v$ have a common distance-preserving spanning tree $T$
  in $G$, a shortest $u$-$v$ path is unique.
\end{lem}
\begin{proof}
  Let $P$ be the $u$-$v$ path in $T$.
  Since $T$ is distance-preserving from $v$, $P$ is a shortest $u$-$v$
  path.
  Assume to the contrary that there is another shortest $u$-$v$ path
  $P_{1}$.
  Then there is a vertex $x$ on $P_{1}$ but not on $P$.

  Let $P_{x}=(u=x_{1},x_{2},\dots,x_{k}=x)$ be the $u$-$x$ path in $T$.
  Let $x_{i}$ be the vertex of $P_{x}$ such that $x_{i-1}$ is on $P$
  but $x_{i}$ is not on $P$.
  Since $P_{x}$ is a shortest $u$-$x$ path, we have
  $d_{T}(u,x)=d_{T}(u,x_{i})+d_{T}(x_{i},x)$.
  Similarly, we obtain $d_{T}(x,v)=d_{T}(x,x_{i})+d_{T}(x_{i},v)$.
  Since $x$ is a vertex on the shortest $u$-$v$ path, we have
  \begin{eqnarray*}
    d_{G}(u,v) &=& d_{G}(u,x) + d_{G}(x,v) \\
    &=& d_{T}(u,x) + d_{T}(x,v) \\
    &=& d_{T}(u,x_{i}) + d_{T}(x_{i},x) + d_{T}(x,x_{i})+d_{T}(x_{i},v) \\
    &=& d_{T}(u,v) + 2 d_{T}(x,x_{i}) \\
    &=& d_{G}(u,v) + 2 d_{T}(x,x_{i}).
  \end{eqnarray*}
  Thus $d_{T}(x,x_{i})=0$, and hence a contradiction is obtained.
\end{proof}

If $u$ and $v$ have a common distance-preserving spanning tree $T$, by
Lemma~\ref{lem:nec_one}, there is a unique shortest $u$-$v$ path
$P=(u=v_{0},v_{1},\dots,v_{k}=v)$.
By the proof of Lemma~\ref{lem:nec_one}, the unique $u$-$v$ path in
$T$ is the unique shortest $u$-$v$ path $P$ in $G$.

\begin{lem}
  \label{lem:nec_two}
  Assume that $u$ and $v$ have a common distance-preserving spanning
  tree $T$ in $G$.
  Let $P=(u=v_{0},v_{1},\dots,v_{k}=v)$ be the unique shortest $u$-$v$
  path in $G$.
  For any vertex $x$ of $G$, there is a unique vertex $v_{i}$ of $P$
  such that $d_{G}(x,v_{i}) = d_{G}(x,V(P))$.
\end{lem}
\begin{proof}
  Let $x$ be a vertex of $G$.
  If $x$ is on $P$, the lemma is trivially true.
  So we assume that $x \not\in V(P)$.
  Let $P_{x}$ be the $u$-$x$ path of $T$.
  Since $T$ is distance-preserving from $u$, $P_{x}$ is a shortest
  path from $u$ to $x$.
  Hence $d_{G}(u,x)=d_{G}(u,w)+d_{G}(w,x)$ for every vertex $w$ of
  $P_{x}$.

  Since $u = v_{0} \in V(P)$ and $x \not\in V(P)$, the path $P_{x}$
  contains a unique vertex $v_{i} \in V(P)$ such that $v_{l} \not\in
  V(P)$ for every $l > i$ (if $P_{x}$ has $v_{k}=v$, then $v_{i}=k$).
  For $0 \leq j \leq i$, we have $d_{T}(u,v_{0}) < d_{T}(u,v_{1}) <
  \dots < d_{T}(u,v_{i})$.
  Since $T$ is distance-preserving from $u$, we have
  \begin{equation}
    \label{eq:2}
    d_{G}(u,v_{0}) < d_{G}(u,v_{1}) < \dots < d_{G}(u,v_{i}).
  \end{equation}
  Since $P_{x}$ is a shortest $u$-$x$ path in $G$, for $0 \leq j \leq
  i$, we have $d_{G}(u,x) = d_{G}(u,v_{j}) + d_{G}(v_{j},x)$.
  Thus $d_{G}(x,v_{j}) = d_{G}(u,x)-d_{G}(u,v_{j})$ for $0 \leq j \leq i$.
  Therefore, by \eqref{eq:2}, we obtain
  \begin{displaymath}
    d_{G}(x,v_{i}) < d_{G}(x,v_{i-1}) < \dots < d_{G}(x,v_{0}) = d_{G}(x,u).
  \end{displaymath}

  Similarly, since $T$ is distance-preserving from $v$, for $i \leq l
  \leq k$, we obtain $d_{G}(x,v_{l}) = d_{G}(v,x)-d_{G}(v_{l},v)$, and
  thus we obtain
  \begin{displaymath}
    d_{G}(x,v_{i}) < d_{G}(x,v_{i+1}) < \dots < d_{G}(x,v_{k}) = d_{G}(x,v).
  \end{displaymath}
  Hence the vertex $v_{i}$ is the unique nearest vertex in $P$ from
  $x$, we obtain $d_{G}(x,v_{i})=d_{G}(x,V(P))$.
\end{proof}

For $0 \leq i \leq k$, we define
\begin{equation}
  \label{eq:1}
  V_{i} = \{ x \mid x \in V(G) \text{ and } d_{G}(x,v_{i}) =
  d_{G}(x,V(P)) \},
\end{equation}
where $P=(u=v_{0},v_{1},\dots,v_{k}=v)$ is the unique shortest $u$-$v$
path defined in Lemma~\ref{lem:nec_two}.

By Lemma~\ref{lem:nec_two}, we can see that $V_{i} \cap V_{j} =
\emptyset$ and $V_{0} \cup V_{1} \cup \dots \cup V_{k} = V(G)$.
That is, $V_{0} \cup V_{1} \cup \dots \cup V_{k}$ is a partition of
$V(G)$.

By the proof of Lemma~\ref{lem:nec_two}, if $x \in V_{i}$, the $u$-$x$
path $P_{x}$ in $T$ contains $v_{i}$, and $P_{x}$ is also a shortest
$u$-$x$ path of $G$.
Hence, for every $x \in V_{i}$, we have $d_{G}(v_{i},x) =
d_{T}(v_{i},x)$.

\begin{lem}
  \label{lem:nec_three}
  Let $V_{0} \cup V_{1} \cup \dots \cup V_{k}$ be a partition defined
  by~\eqref{eq:1}.
  If $e = xy \in E(G)$ for $x \in V_{i}$ and $y \in V_{j}$, then $w(e)
  \geq d_{G}(v_{i},v_{j})$ and $|d_{G}(v_{i},x)-d_{G}(v_{j},y)| \leq
  w(e)-d_{G}(v_{i},v_{j})$.
\end{lem}
\begin{proof}
  If $i=j$, the lemma is true.
  Hence we assume that $i < j$.

  Since $x \in V_{i}$ is adjacent to $y \in V_{j}$, we obtain
  \begin{equation}
    \label{eq:6}
    d_{G}(u,y) \leq d_{G}(u,x) + w(e).
  \end{equation}
  Since $T$ is distance-preserving from $u$, we have
  \begin{align*}
    d_{T}(u, x) &= d_{T}(u, v_{i}) + d_{T}(v_{i}, x), \\
    d_{T}(u, y) &= d_{T}(u, v_{j}) + d_{T}(v_{j}, y).
  \end{align*}
  Thus, by \eqref{eq:6}, $d_{T}(u, v_{j}) + d_{T}(v_{j}, y) \leq  d_{T}(u, v_{i}) +
  d_{T}(v_{i}, x) + w(e)$.
  Since $d_{T}(u,v_{j}) -d_{T}(u,v_{i}) = d_{T}(v_{i},v_{j}) =
  d_{G}(v_{i},v_{j})$, we obtain
  \begin{equation}
    \label{eq:7}
    d_{G}(v_{i},v_{j}) \leq d_{T}(v_{i}, x) - d_{T}(v_{j}, y) + w(e).
  \end{equation}
  Similarly, by considering the fact that $T$ is distance-preserving
  from $v$, we obtain
  \begin{equation}
    \label{eq:5}
    d_{G}(v_{i},v_{j}) \leq d_{T}(v_{j}, y) - d_{T}(v_{i}, x) + w(e).
  \end{equation}
  By adding the both side of inequalities \eqref{eq:7} and
  \eqref{eq:5}, we have $d_{G}(v_{i},v_{j}) \leq w(e)$.

  From \eqref{eq:7} and the fact $d_{G}(v_{i},x)=d_{T}(v_{i},x)$, we
  obtain
  \begin{displaymath}
    d_{G}(v_{i}, x) - d_{G}(v_{j}, y) \geq  -(w(e) - d_{G}(v_{i},v_{j})).
  \end{displaymath}
  Similarly, from \eqref{eq:5},
  \begin{displaymath}
    d_{T}(v_{i}, x) - d_{T}(v_{j}, y) \leq w(e) - d_{G}(v_{i},v_{j}).
  \end{displaymath}
  Thus we obtain
  \begin{displaymath}
    |d_{G}(v_{i}, x) - d_{G}(v_{j}, y)| \leq w(e) - d_{G}(v_{i},v_{j}).
  \end{displaymath}
\end{proof}

By Lemmas~\ref{lem:nec_one}, \ref{lem:nec_two} and
\ref{lem:nec_three}, we have shown the necessary condition in
Theorem~\ref{thm:main}.



Next we prove the sufficiency of Theorem~\ref{thm:main}.
We assume that two vertices $u$ and $v$ in $G$ satisfy the following
three conditions.

\begin{enumerate}
\item[(1)] A shortest $u$-$v$ path $P$ in $G$ is unique.

\item[(2)] We define the shortest $u$-$v$ path as
  $P=(u=v_{0},v_{1},\dots,v_{k}=v)$.
  For any vertex $x$, there is a unique vertex $v_{i} \in V(P)$ such
  that $d_{G}(x,v_{i})=d_{G}(x,V(P))$.

\item[(3)] For $0 \leq i \leq k$, let $V_{i}=\{x \mid x \in V(G) \text{
    and } d_{G}(x,v_{i})=d_{G}(x,V(P))\}$.
  If $e = xy \in E(G)$ for $x \in V_{i}$ and $y \in V_{j}$, then
  $w(e) \geq d_{G}(v_{i},v_{j})$ and
  $|d_{G}(v_{i},x)-d_{G}(v_{j},y)| \leq w(e)-d_{G}(v_{i},v_{j})$.
\end{enumerate}

For $0 \leq i \leq k$, let $G_{i}$ be the subgraph of $G$ induced by
$V_{i}$.

\begin{lem}
  \label{lem:suf_con}
  For $0 \leq i \leq k$, the induced subgraph $G_{i}=G[V_{i}]$ is
  connected.
\end{lem}
\begin{proof}
  Assume that $G_{i}$ is disconnected for some $i$.
  Let $x$ be a vertex in a component that does not contain $v_{i}$.
  So, a shortest $x$-$v_{i}$ path $P_{x}$ of $G$ have to contain an
  edge $e=yw$ such that $y \in V_{i}$ and $w \in V_{j}$ for $j \neq
  i$.
  Since $P_{x}$ is a shortest $x$-$v_{i}$ path $G$, we have
  $d_{G}(x,v_{i}) = d_{G}(x,w) + d_{G}(w,v_{i})$.
  By the definition of $V_{i}$, we have $d_{G}(w,v_{j}) <
  d_{G}(w,v_{i})$.
  Hence, we obtain
  \begin{align*}
    d_{G}(x,v_{i})
    &= d_{G}(x,w) + d_{G}(w,v_{i}) \\
    &> d_{G}(x,w) + d_{G}(w,v_{j}) \\
    &\geq d_{G}(x,v_{j}).
  \end{align*}
  This contradicts the fact that $x \in V_{i}$.
\end{proof}

\begin{lem}
  \label{lem:suf_dist}
  For $0 \leq i \leq k$ and any vertex $x \in V_{i}$,
  \begin{displaymath}
    d_{G_{i}}(v_{i}, x) = d_{G}(v_{i}, x).
  \end{displaymath}
\end{lem}
\begin{proof}
  Since $G_{i}$ is a connected subgraph of $G$, clearly
  $d_{G_{i}}(v_{i}, x) \geq d_{G}(v_{i}, x)$.
  Assume that there is a vertex $x \in V_{i}$ such that
  $d_{G_{i}}(v_{i}, x) > d_{G}(v_{i}, x)$.

  In this case, a shortest $v_{i}$-$x$ path contains a vertex $y \in
  V_{j}$ and $j \neq i$.
  Hence
  \begin{align*}
    d_{G}(v_{i},x) &= d_{G}(v_{i}, y) + d_{G}(y, x) \\
    &> d_{G}(v_{j}, y) + d_{G}(y, x) \\
    &\geq d_{G}(v_{j},x).
  \end{align*}
  This contradicts the fact that $x \in V_{i}$.
\end{proof}

Now we are ready to prove the sufficiency of Theorem~\ref{thm:main}.

\begin{proof}[Proof of Sufficiency]
  By Lemma~\ref{lem:suf_con}, $G_{i}$ is connected.
  So, $G_{i}$ has a distance-preserving spanning tree $T_{i}$ from
  $v_{i}$.
  We define a spanning tree $T$ of $G$ that has the edge set
  \begin{equation}
    \label{eq:9}
    E(T) = E(P) \cup E(T_{0}) \cup E(T_{1}) \cup \dots \cup E(T_{k}),
  \end{equation}
  where $P$ is the unique shortest $u$-$v$ path of $G$.
  We can see easily that $T$ is a spanning tree of $G$.

  We show that the tree $T$ is a common distance-preserving spanning
  tree of $u$ and $v$.
  That is, for any vertex $x$, we show that $d_{T}(u,x) =
  d_{G}(u,x)$ and $d_{T}(v,x) = d_{G}(v,x)$.
  In this proof, we show that $T$ is distance-preserving from $u$.
  We can prove similarly $T$ is distance-preserving from $v$.

  For a vertex $x$, mappings $p$ and $h$ are defined as
  \begin{displaymath}
    \begin{cases}
      p(x) = d_{G}(u,v_{i}), & \text{when $x \in V_{i}$},\\
      h(x) = d_{G}(v_{i},x), & \text{when $x \in V_{i}$},
    \end{cases}
  \end{displaymath}
  and then define $W(x) = p(x) + h(x)$.
  It is easy to see that $d_{T}(u,x) = p(x) + h(x)$ for any $x$.
  By the definition, $W(u_{0})=0+0=0$.
  Let $P_{x}=(u=u_{0},u_{1},\dots,u_{s}=x)$ be a shortest $u$-$x$ path
  in $G$.
  Since $P$ is a shortest path, we have
  $d_{G}(u,u_{i+1})=d_{G}(u,u_{i}) + w(e_{i})$, where
  $e_{i}=u_{i}u_{i+1}$.
  Consider the sequence $W(u_{0}),W(u_{1}),\dots,W(u_{s})$ and the
  value of $|W(u_{i+1}) - W(u_{i})|$.

  We first assume that the edge $e_{i}=u_{i}u_{i+1}$ is a edge of $T$.
  If $e_{i}$ is in $P$, $|W(u_{i+1}) - W(u_{i})| =
  |p(u_{i+1})-p(u_{i})| = w(e_{i})$.
  If $e_{i}$ is $T_{i}$, $|W(u_{i+1}) - W(u_{i})| =
  |h(u_{i+1})-h(u_{i})| = w(e_{i})$.
  Thus we have $|W(u_{i+1}) - W(u_{i})| = w(e_{i})$ when $e_{i}$ is in
  $T$.

  Next we suppose that $e_{i}$ is not in $T$.
  If $u_{i} \in V_{j}$ and $u_{i+1} \in V_{j'}$, by the
  condition~(3), we obtain
  \begin{align*}
    |W(u_{i+1}) - W(u_{i})|
    &= |((p(u_{i+1})-p(u_{i})) + (h(u_{i+1})-h(u_{i}))| \\
    &= |d_{G}(v_{j'}, v_{j}) + (d_{G}(v_{j'},u_{i+1}) - d_{G}(v_{j},u_{i}))| \\
    &\leq d_{G}(v_{j'}, v_{j}) + w(e_{i}) - d_{G}(v_{j'},v_{j}) \text{ (by condition~(3))} \\
    &= w(e_{i}).
  \end{align*}
  In both cases, we obtain $|W(u_{i+1}) - W(u_{i})| \leq w(e_{i})$.
  Hence,
  \begin{align*}
    d_{T}(u,x)
    &= W(u_{s}) - W(u_{0}) \\
    &= \sum_{0 \leq i \leq s-1} (W(u_{i+1}) - W(u_{i})) \\
    &\leq \sum_{0 \leq i \leq s-1} w(e_{i}) \\
    &= d_{G}(u,x).
  \end{align*}
  Since $T$ is a connected subgraph of $G$, we have $d_{T}(u,x) \geq
  d_{G}(u,x)$.
  Thus, we obtain $d_{T}(u,x) = d_{G}(u,x)$ for any vertex $x$.
\end{proof}

We have completed the proof of Theorem~\ref{thm:main}.



\end{document}